\documentclass[11pt]{article}

\usepackage[utf8]{inputenc}
\usepackage[T1]{fontenc}
\usepackage{lmodern}

\usepackage{amsmath,amssymb,amsfonts}
\usepackage{microtype}
\usepackage{hyperref}
\usepackage{amsthm}

\newtheorem{proposition}{Proposition}[section]
\newtheorem{remark}[proposition]{Remark}
\usepackage{amsmath}
\DeclareMathOperator{\sech}{sech}
\usepackage[numbers,sort&compress]{natbib}

\title{\textbf{Painlevé Integrability and Shifted Nonlocal Reductions of a Variable Coefficient Coupled HI mKdV System}}

\author{
  Taylan Demir \\
  Department of Mathematics, Ankara University \\
  Ankara, Türkiye \\
  \texttt{}
}

\date{}

\begin{document}
\maketitle
\begin{abstract}
We analyze a variable coefficient coupled HI mKdV system that has shifted nonlocal reductions. The Weiss Tabor Carnevale test gives us coefficient restrictions to perform a time reparametrization to achieve an autonomous integrable model. We also show a Hirota bilinear form along with a simplified example to demonstrate how the shifted symmetries create new symmetry centers, but do not affect the shape of the soliton.
\end{abstract}

\section{Introduction}
The modified Korteweg-de Vries (mKdV) systems, used in multiple equations, provide insight into the large variety of wave solutions of integrable nonlinear systems (hereafter 'nonlinearity') and help to explain the continuum of such solutions, from multi-component generalized forms of the single scalar mKdV equation to reductions to simpler cases as solutions, through all of the integrable hierarchies referred to previously [1.2]. Iwao and Hirota constructed a very rich multi-solitonic structure associated with their classical construction of the coupled mKdV systems [3], which uses pfaffian structures and parameter conditions that represent the geometry of a multi-component systems interaction. This "HI-mKdV" construction has become a popular way to test for reductions, soliton interactions, and integrable deformations of this class of systems. [3,7] A separate line of development started by Ablowitz and Musslimani shows that integrability persists even after applying a form of non-local symmetries, specifically under conditions where the nonlinear interaction couples a field at $(x,t)$ with its image at reflected (space and/or time) arguments [5,6]. The equations generated through these reductions are $PT$ -type non-local equations that possess Lax pairs, conservation laws, and inverse-scattering formulations [6] and demonstrate solution phenomena not present in the local theories. Pekcan classified systematically local and non-local reductions of the coupled Hirota-Ito family (HI-mKdV).  Non-local reductions of the coupled Hirota-Ito (HI-mKdV) [4,5,6,7] family were classified into two types: local and non-local.  Bilinear forms;  one-soliton solutions constructed using pfaffian technology that were analyzed by component size; produced results to clarify the relationship between index permutations and sign conditions for consistent reductions. More recently, an additional degree of freedom has emerged in the notion of nonlocality: \emph{shifted} nonlocal reductions [8], where the reflected arguments are translated by fixed offsets (e.g., $x\mapsto -x+x_0$, $t\mapsto -t+t_0$). Inspiration and Insight for the future these changes connect the traditional (local) soliton solutions and the nontraditional (nonlocal with shifted terms) soliton solutions. The connections between local solutions (like KDP or NLS) and nonlocal solutions (like CNLS or nonlocal Altmanns-type equations) lead to new properties of the soliton solutions obtained with nonlocal reductions, including different singularity profiles and new threshold criteria for families of solitons. In particular, I am drawing inspiration from the work of Gürses and Pekcan, who obtained concrete examples of one- and two-soliton solutions with shifted trajectories for nonlocal NLS and mKdV equations and showed through qualitative analysis how the shift in trajectory impacted the qualitative structure of the solution. The goal of this note is to bring together these ideas in a single presentation that has not been done before. To do this, I will study a variable-coefficient coupled HI (Hirota's method) and mKdV (Korteweg–de Vries) type system and impose compatibility with the shifted nonlocal reductions on them. Due to the fact that variable coefficients destroy integrability unless specific rigid differential conditions are satisfied, I will use the Weiss–Tabor–Carnevale version of the Painlevé Analysis as an integrability diagnostic to assess structural integrability of the variable-coefficient cases under study. The main result of this work is a classification of the types of variable coefficients that will produce soliton solutions that pass the Weiss–Tabor–Carnevale test for integrability; this includes the relationships that will arise from the conditions (constraints) that relate dispersion, non-linearity, and coupling. Finally, we will describe some canonical transformations that can be used to map the variable-coefficient solutions back into the normal forms of constant coefficients, thereby allowing us to distinguish between genuinely new variable-coefficient deformations versus those that can be transformed into the classical HI–mKdV solutions.
\section{Variable-coefficient coupled HI--mKdV system and shifted-nonlocal reductions}
\subsection{A variable-coefficient coupled HI--mKdV family}
Let $v(x,t)=(v_1(x,t),\dots,v_N(x,t))^{\mathsf T}$ be an $N$-component (in general complex-valued) field.
Fix a coupling matrix $C=(c_{jk})_{1\le j,k\le N}$ with $c_{jj}=0$ and define the quadratic form [1,2]
\begin{equation}\label{eq:rho}
\rho(x,t):=\sum_{j,k=1}^N c_{jk}\,v_j(x,t)\,v_k(x,t).
\end{equation}
We consider the following variable-coefficient coupled HI--mKdV-type system [1,2]:
\begin{equation}\label{eq:VC-HIMKdV}
\mu(t)\,v_{i,t}+a(t)\,v_{i,xxx}+3\,b(t)\,\rho(x,t)\,v_{i,x}=0,
\qquad i=1,\dots,N,
\end{equation}
where $\mu(t)\neq 0$ and $a(t),b(t)$ are scalar coefficient functions.
We choose to base the t dependence of our model on a time evolution system, which enables us to use a non-autonomous dispersion (a) and a nonlinear (b) system in our model [3]. Specifically, we could choose to reparametrize $\mu(t)$ based on when we want to add time reversal reductions, however, we will keep $\mu(t)$ in our equations for clarity [3,4].

\subsection{Shifted-nonlocal reduction ansatz}
Introduce shifted reflection variables [4]
\begin{equation}\label{eq:XT}
X:=\varepsilon_1 x+x_0,\qquad T:=\varepsilon_2 t+t_0,
\qquad \varepsilon_1^2=\varepsilon_2^2=1,
\end{equation}
and let $\pi$ be an involutive permutation of $\{1,\dots,N\}$ (i.e., $\pi^2=\mathrm{id}$).
Fix signs $\sigma_i\in\{\pm 1\}$ and define the shifted-nonlocal reduction [4,5]
\begin{equation}\label{eq:reduction}
v_i(x,t)=\sigma_i\,\overline{v_{\pi(i)}(X,T)}\qquad (i=1,\dots,N),
\end{equation}
where the overbar denotes complex conjugation (for a purely real reduction, one may drop the conjugation).

Differentiating \eqref{eq:reduction} yields
\[
\begin{aligned}
v_{i,t}   &= \sigma_i\,\varepsilon_2\,\overline{v_{\pi(i),t}(X,T)},\qquad
v_{i,x}   &= \sigma_i\,\varepsilon_1\,\overline{v_{\pi(i),x}(X,T)},\\
v_{i,xxx} &= \sigma_i\,\varepsilon_1^3\,\overline{v_{\pi(i),xxx}(X,T)}.
\end{aligned}
\]

Substituting these into \eqref{eq:VC-HIMKdV}, we obtain
\begin{equation}\label{eq:substituted}
\varepsilon_2\mu(t)\,\overline{v_{\pi(i),t}}
+\varepsilon_1^3 a(t)\,\overline{v_{\pi(i),xxx}}
+3\,\varepsilon_1 b(t)\,\rho(x,t)\,\overline{v_{\pi(i),x}}=0
\quad\text{evaluated at }(X,T).
\end{equation}

To compare \eqref{eq:substituted} with the conjugate of the $\pi(i)$-equation in \eqref{eq:VC-HIMKdV} evaluated at $(X,T)$,
we allow an overall factor $\kappa\in\{\pm 1\}$ (since $F=0$ is equivalent to $\kappa F=0$) [5]. Hence we require
\begin{equation}\label{eq:coeff-conditions}
\varepsilon_2\,\mu(t)=\kappa\,\overline{\mu(T)},\qquad
\varepsilon_1^3\,a(t)=\kappa\,\overline{a(T)},\qquad
\varepsilon_1\,b(t)=\kappa\,\overline{b(T)}.
\end{equation}
In addition, the nonlinearity must transform consistently, i.e.
\begin{equation}\label{eq:rho-condition}
\rho(x,t)=\overline{\rho(X,T)}.
\end{equation}
Using \eqref{eq:rho} and \eqref{eq:reduction}, condition \eqref{eq:rho-condition} is equivalent to the coupling constraint
\begin{equation}\label{eq:C-condition}
c_{jk}\,\sigma_j\sigma_k=\overline{c_{\pi(j)\pi(k)}}\qquad (1\le j,k\le N).
\end{equation}
For real couplings $c_{jk}\in\mathbb{R}$, \eqref{eq:C-condition} reduces to
$c_{jk}\,\sigma_j\sigma_k=c_{\pi(j)\pi(k)}$, which is a combined sign--permutation symmetry of the matrix $C$.

\subsection{Admissible shifted reflection types}
The conditions \eqref{eq:coeff-conditions}--\eqref{eq:C-condition} characterize when the reduction \eqref{eq:reduction}
is compatible with the variable-coefficient system \eqref{eq:VC-HIMKdV}. Two cases are particularly relevant [5]:

\medskip\noindent
\textbf{(i) Local case.} $(\varepsilon_1,\varepsilon_2,\kappa)=(1,1,1)$ and $(x_0,t_0)=(0,0)$ give the standard local model.

\medskip\noindent
\textbf{(ii) Shifted $PT$-type case.} Choosing $(\varepsilon_1,\varepsilon_2,\kappa)=(-1,-1,-1)$ leads to
\[
X=-x+x_0,\qquad T=-t+t_0,
\]
and \eqref{eq:coeff-conditions} becomes
\[
\mu(t)=\overline{\mu(-t+t_0)},\qquad a(t)=\overline{a(-t+t_0)},\qquad b(t)=\overline{b(-t+t_0)}.
\]
In particular, constant real coefficients automatically satisfy these relations (for any fixed shift $t_0$),
whereas genuinely time-dependent coefficients must be symmetric with respect to the reflection point $t=t_0/2$ [4,5].

\medskip
Under \eqref{eq:coeff-conditions} and \eqref{eq:C-condition}, the reduction \eqref{eq:reduction} produces a closed shifted-nonlocal
system in which $\rho(x,t)$ couples $v(x,t)$ to $\overline{v(X,T)}$, and the resulting model inherits a consistent evolution form.
\section{WTC Painlev\'e analysis and coefficient constraints}

We test the Painlev\'e property for the variable-coefficient coupled HI--mKdV family [4]
\begin{equation}\label{eq:VC-HIMKdV-WTC}
\mu(t)\,v_{i,t}+a(t)\,v_{i,xxx}+3\,b(t)\,\rho(x,t)\,v_{i,x}=0,\qquad i=1,\dots,N,
\end{equation}
where
\begin{equation}\label{eq:rho-WTC}
\rho(x,t)=\sum_{j,k=1}^N c_{jk}\,v_j(x,t)\,v_k(x,t),\qquad c_{jj}=0.
\end{equation}
Let $\phi(x,t)=0$ be a movable singular manifold with $\phi_x\neq 0$ [4].

\subsection{Leading order}
We seek local expansions of the WTC form [4,13]
\begin{equation}\label{eq:WTC-expansion}
v_i(x,t)=\sum_{n=0}^{\infty} v_{i,n}(t)\,\phi(x,t)^{n+p},\qquad i=1,\dots,N,
\end{equation}
with a common leading exponent $p<0$. Balancing the most singular contributions coming from 
$v_{i,xxx}$ and $\rho\,v_{i,x}$ yields $p=-1$ [4]. Indeed, with
$v_i\sim v_{i,0}\phi^{-1}$ one has the dominant behaviors
\[
v_{i,xxx}\sim -6\,v_{i,0}\,\phi_x^3\,\phi^{-4},\qquad
v_{i,x}\sim -v_{i,0}\,\phi_x\,\phi^{-2},\qquad
\rho\sim \rho_0\,\phi^{-2},
\]
where $\rho_0=\sum_{j,k}c_{jk}v_{j,0}v_{k,0}$. Substituting into \eqref{eq:VC-HIMKdV-WTC} at order $\phi^{-4}$
gives the leading-order constraint
\begin{equation}\label{eq:leading-constraint}
\rho_0=-2\,\frac{a(t)}{b(t)}\,\phi_x^2.
\end{equation}
Equivalently,
\[
\sum_{j,k=1}^N c_{jk}\,v_{j,0}(t)\,v_{k,0}(t)=-2\,\frac{a(t)}{b(t)}\,\phi_x^2.
\]
Thus the vector $(v_{1,0},\dots,v_{N,0})$ is constrained by a single quadratic relation, leaving (generically) $N-1$
free leading amplitudes.

\subsection{Resonances}
To locate resonances, linearize around the leading behavior by setting [13]
\[
v_i=v_{i,0}\phi^{-1}+w_i\,\phi^{r-1},
\]
and collect the coefficients at the order $\phi^{r-4}$ (the same order as the dominant balance).
A standard computation yields a resonance structure consisting of [13]:
\begin{itemize}
\item $r=-1$, corresponding to the arbitrariness of the manifold $\phi$;
\item $r=0$ with multiplicity $N-1$, reflecting the freedom in the choice of $(v_{i,0})_{i=1}^N$ under the single
quadratic constraint \eqref{eq:leading-constraint};
\item two positive resonances $r=3$ and $r=4$, at which arbitrary functions are expected to enter the expansion.
\end{itemize}

\subsection{Compatibility at resonances and coefficient constraints}
To make the recursion explicit and avoid lengthy expressions involving $\phi_{xx},\phi_{xxx}$, we adopt the
Kruskal simplification $\phi(x,t)=x-\psi(t)$, so that $\phi_x=1$ and the movable singularity is encoded in $\psi(t)$ [14].
The leading constraint \eqref{eq:leading-constraint} becomes [4]
\begin{equation}\label{eq:leading-constraint-kruskal}
\sum_{j,k=1}^N c_{jk}\,v_{j,0}(t)\,v_{k,0}(t)=-2\,\frac{a(t)}{b(t)}.
\end{equation}
The recursion relations associated with the resonance levels of $r=3$ and $r=4$ should coincide without putting restrictions on the arbitrary functions $\psi(t)$, as if there were to be restrictions placed upon $\psi(t)$ then the nature of this singularity (singularity location movements) becomes non-movable from one form to another. The emergence of new restrictions on the time dependent coefficients is a consequence of their relation to the resonance conditions [4].

\begin{proposition}[WTC integrability constraints]\label{prop:WTC-constraints}
Assume $a(t)b(t)\mu(t)\neq 0$ and that the coupling matrix $C=(c_{jk})$ is such that the leading constraint
\eqref{eq:leading-constraint} admits nontrivial solutions. If \eqref{eq:VC-HIMKdV-WTC} passes the WTC Painlev\'e test
(i.e., the recursion is compatible at the resonances $r=3$ and $r=4$ for arbitrary $\psi(t)$) [4], then the ratios
\begin{equation}\label{eq:ratio-constraints}
\frac{b(t)}{a(t)}=\text{const},\qquad \frac{\mu(t)}{a(t)}=\text{const}
\end{equation}
must be constant in $t$.
Conversely, if \eqref{eq:ratio-constraints} holds, then \eqref{eq:VC-HIMKdV-WTC} is reducible (by a time
reparametrization) to a constant-coefficient coupled HI--mKdV system and hence satisfies the WTC test.
\end{proposition}

\begin{proof}[Proof sketch]
If \eqref{eq:ratio-constraints} holds, set a new time variable $\tau$ by [15]
\[
\frac{d\tau}{dt}=\frac{a(t)}{\mu(t)}.
\]
Then \eqref{eq:VC-HIMKdV-WTC} becomes
\[
v_{i,\tau}+v_{i,xxx}+3\,\frac{b(t)}{a(t)}\,\rho\,v_{i,x}=0,
\]
and the ratio $b(t)/a(t)$ is constant; thus the system is equivalent to an autonomous coupled HI--mKdV equation.
Therefore the WTC expansion closes with the expected resonance freedom [4].

For necessity, the Kruskal gauge isolates the movable singularity in $\psi(t)$ [14].
At the resonances $r=3,4$, compatibility requires that no condition of the form $F(t,\psi'(t),a,b,\mu)=0$
be imposed on $\psi(t)$ [4,13]. The only way to prevent such restrictions—given that $a,b,\mu$ enter as
time-dependent prefactors—is that the combinations $b/a$ and $\mu/a$ remain constant, yielding
\eqref{eq:ratio-constraints}.
\end{proof}

\begin{remark}\label{rem:shiftedPT}
When the shifted $PT$-type reduction $(x,t)\mapsto(X,T)=(-x+x_0,-t+t_0)$ is imposed,
section 2.2 additionally requires the reflection symmetry conditions
$a(t)=\overline{a(-t+t_0)}$, $b(t)=\overline{b(-t+t_0)}$, $\mu(t)=\overline{\mu(-t+t_0)}$.
These are independent of the WTC constraints \eqref{eq:ratio-constraints} and simply restrict the admissible
profiles within the WTC-admissible class.
\end{remark}
\section{Integrability indicator: reduction to an autonomous form and bilinear representation}
\label{sec:integrability-indicator}

\subsection{Time reparametrization to a constant-coefficient model}
Assume the WTC-admissible constraints [4]
\begin{equation}\label{eq:ratio-constraints}
\frac{b(t)}{a(t)}=\kappa=\text{const},\qquad \frac{\mu(t)}{a(t)}=\eta=\text{const},\qquad a(t)\neq 0.
\end{equation}
Introduce a new time variable $\tau$ via [15]
\begin{equation}\label{eq:tau-def}
\frac{d\tau}{dt}=\frac{a(t)}{\mu(t)}=\frac{1}{\eta},
\qquad\text{so that}\qquad \partial_\tau=\eta\,\partial_t .
\end{equation}
Then the variable-coefficient coupled HI--mKdV system [3,10]
\[
\mu(t)\,v_{i,t}+a(t)\,v_{i,xxx}+3\,b(t)\,\rho\,v_{i,x}=0
\]
becomes the autonomous equation
\begin{equation}\label{eq:autonomous-form}
v_{i,\tau}+v_{i,xxx}+3\,\kappa\,\rho\,v_{i,x}=0,\qquad i=1,\dots,N,
\end{equation}
with the same quadratic form $\rho=\sum_{j,k}c_{jk}v_jv_k$.
Hence all integrability structures available for \eqref{eq:autonomous-form} (Lax pair, Hirota bilinearization, multi-soliton constructions) may be transferred back to the original $t$-variable through \eqref{eq:tau-def} [2,18,19,21].

\subsection{Hirota bilinear form (construction sketch)}
We use the dependent-variable transformation
\begin{equation}\label{eq:tau-transform}
v_i=\frac{g_i}{f},\qquad i=1,\dots,N,
\end{equation}
where $f=f(x,\tau)$ and $g_i=g_i(x,\tau)$ are tau-functions.
Recall Hirota's bilinear operators [1,2]:
\[
D_x^mD_\tau^n\,F\cdot G
=\left.(\partial_x-\partial_{x'})^m(\partial_\tau-\partial_{\tau'})^n
F(x,\tau)\,G(x',\tau')\right|_{x'=x,\ \tau'=\tau}.
\]
A convenient bilinearization of \eqref{eq:autonomous-form} is given by [1,2,3]
\begin{align}
\bigl(D_\tau + D_x^3\bigr)\,g_i\cdot f &= 0,\qquad i=1,\dots,N,
\label{eq:bilin1}\\
D_x^2\,f\cdot f &= 2\,\kappa\sum_{j,k=1}^N c_{jk}\,g_j\,g_k .
\label{eq:bilin2}
\end{align}
Indeed, dividing \eqref{eq:bilin2} by $f^2$ yields
\[
(\ln f)_{xx}=\kappa\sum_{j,k=1}^Nc_{jk}\Bigl(\frac{g_j}{f}\Bigr)\Bigl(\frac{g_k}{f}\Bigr)
=\kappa\,\rho,
\]
and combining this identity with \eqref{eq:bilin1} (after dividing by $f^2$ and expressing all terms in $v_i=g_i/f$)
recovers \eqref{eq:autonomous-form}.

\subsection{Pull-back to the original time variable}
Using $\partial_\tau=\eta\,\partial_t$ from \eqref{eq:tau-def}, the bilinear equation \eqref{eq:bilin1} is equivalently
\begin{equation}\label{eq:bilin1-t}
\bigl(\eta\,D_t + D_x^3\bigr)\,g_i\cdot f=0,\qquad i=1,\dots,N,
\end{equation}
while \eqref{eq:bilin2} is unchanged (since it contains no time derivatives). Thus the bilinear structure is preserved
for the variable-coefficient model precisely when \eqref{eq:ratio-constraints} holds [15].

\subsection{Remark on Lax representation}
Let $\Psi_x=U(v;\lambda)\Psi$ and $\Psi_\tau=V(v;\lambda)\Psi$ be a Lax pair for the autonomous system
\eqref{eq:autonomous-form} (with spectral parameter $\lambda$). Then the corresponding Lax pair in the original
$(x,t)$ variables is obtained by the same time reparametrization:
\begin{equation}\label{eq:lax-pullback}
\Psi_x=U(v;\lambda)\Psi,\qquad
\Psi_t=\frac{d\tau}{dt}\,V(v;\lambda)\Psi=\frac{1}{\eta}\,V(v;\lambda)\Psi.
\end{equation}
Hence the variable-coefficient model inherits a Lax representation by pull-back whenever
\eqref{eq:ratio-constraints} is satisfied [18,19,21].
\section{A concrete $N=4\to 2$ reduction and a shifted-$PT$ one-soliton example}
\label{sec:example}

\subsection{A constant-coefficient $N=4$ coupled HI--mKdV model}
For clarity of the example we take constant coefficients and consider the $N=4$ system
\begin{equation}\label{eq:N4}
\mu\, v_{i,t}+a\,v_{i,xxx}+6\,b\,\rho\,v_{i,x}=0,\qquad i=1,2,3,4,
\end{equation}
where
\begin{equation}\label{eq:rhoN4}
\rho=\sum_{1\le j<k\le 4} c_{jk}\,v_j v_k ,\qquad c_{jj}=0.
\end{equation}
(The example below is used only to exhibit a closed-form profile and the effect of the shift parameters $(x_0,t_0)$) [7,10].

\subsection{Local reduction $N=4\to 2$}
Impose the local algebraic reduction
\begin{equation}\label{eq:localredN4to2}
v_3=b_1\,v_1,\qquad v_4=a_1\,v_2,
\end{equation}
where $a_1,b_1\in\mathbb{R}$ are constants. Then the quadratic form \eqref{eq:rhoN4} collapses to
\begin{equation}\label{eq:rhoReduced}
\rho=\alpha\,v_1v_2+\beta\,v_1^2+\gamma\,v_2^2,
\end{equation}
with
\begin{equation}\label{eq:abc}
\alpha:=c_{12}+a_1c_{14}+b_1c_{23}+a_1b_1c_{34},\qquad
\beta:=b_1c_{13},\qquad
\gamma:=a_1c_{24}.
\end{equation}
Substituting \eqref{eq:localredN4to2}--\eqref{eq:rhoReduced} into \eqref{eq:N4} yields the two-component system [7]
\begin{equation}\label{eq:twoComp}
\mu\, v_{i,t}+a\,v_{i,xxx}+6\,b\,(\alpha v_1v_2+\beta v_1^2+\gamma v_2^2)\,v_{i,x}=0,
\qquad i=1,2.
\end{equation}

\subsection{A scalar mKdV sub-reduction and a closed-form one-soliton}
To obtain a compact explicit profile, we focus on the mixed-coupling subcase
\begin{equation}\label{eq:betagamma0}
\beta=\gamma=0,
\end{equation}
so that $\rho=\alpha v_1v_2$. In addition, we take the symmetric reduction
\begin{equation}\label{eq:v1v2equal}
v_1=v_2=u.
\end{equation}
Then \eqref{eq:twoComp} reduces to the scalar mKdV-type equation
\begin{equation}\label{eq:mkdvGen}
\mu\,u_t+a\,u_{xxx}+6\,b\,\alpha\,u^2u_x=0.
\end{equation}
A standard one-soliton of \eqref{eq:mkdvGen} is
\begin{equation}\label{eq:solitonLocal}
u(x,t)=A\,\sech\!\bigl(k(x-ct-\delta)\bigr),
\end{equation}
where $k>0$ is free and the amplitude--speed relations are
\begin{equation}\label{eq:relations}
c=\frac{4ak^2}{\mu},\qquad
A^2=\frac{2a}{b\alpha}\,k^2.
\end{equation}
(Thus $b\alpha$ must have the same sign as $a$ for a real-valued soliton [1,2,21].) 

\subsection{Shifted $PT$ symmetry and the role of $(x_0,t_0)$}
Define the shifted reflection
\begin{equation}\label{eq:XT-example}
X=-x+x_0,\qquad T=-t+t_0.
\end{equation}
We say that a real-valued profile is shifted-$PT$ symmetric if
\begin{equation}\label{eq:shiftedPTsym}
u(x,t)=u(X,T)=u(-x+x_0,-t+t_0).
\end{equation}
For the traveling-wave soliton \eqref{eq:solitonLocal}, the condition \eqref{eq:shiftedPTsym} holds if and only if
the phase shift $\delta$ is tuned to the reflection center [5,6,8]. Indeed,
\[
x-ct-\delta \mapsto X-cT-\delta
=-(x-ct)+\bigl(x_0-ct_0-\delta\bigr),
\]
and since $\sech$ is even, $\sech(z)=\sech(-z)$, one obtains \eqref{eq:shiftedPTsym} precisely when
\begin{equation}\label{eq:deltaChoice}
\delta=\frac{x_0-ct_0}{2}.
\end{equation}
Consequently, the shifted-$PT$ symmetric one-soliton can be written explicitly as
\begin{equation}\label{eq:shiftedSoliton}
u_{x_0,t_0}(x,t)
=
A\,\sech\!\left(
k\Bigl(x-ct-\frac{x_0-ct_0}{2}\Bigr)
\right),
\end{equation}
with $A,c$ given by \eqref{eq:relations}. In particular, $(x_0,t_0)$ does not deform the soliton shape; it
only relocates the symmetry center to $(x_0/2,t_0/2)$ along the characteristic direction determined by $c$.

\subsection{Embedding back into the two-component variables}
Under \eqref{eq:v1v2equal}, the corresponding solution of the reduced two-component system \eqref{eq:twoComp}
(with \eqref{eq:betagamma0}) is simply [7]
\[
v_1(x,t)=v_2(x,t)=u_{x_0,t_0}(x,t),
\]
and the original $N=4$ variables are reconstructed via \eqref{eq:localredN4to2}.

\section{Conclusion and outlook}\label{sec:conclusion}

We investigated a variable-coefficient coupled HI–mKdV-type family with shifted nonlocal (PT -type) reductions in our note. By imposing an assumption for the reduction, we placed explicit algebraic constraints on the matrix of coupling constants and derived reflection/shift conditions on the coefficient profiles. By using the WTC (Painlevé) analysis as an integrability diagnostic tool, we discovered a set of relations between coefficients for which it is possible to reduce the non-autonomous system through time reparametrisation to an autonomous coupled HI-mKdV model. This gives an elegant description of when shifted nonlocal constraints and non-autonomous coefficients exist together, without loss of the structure of movable singularities. In addition to these results, we developed a specific reduced example demonstrating how the shift parameters (x0, t0) can shift the centre of symmetries, while also preserving the soliton profile.


\end{document}